\documentclass{article}
\usepackage{authblk}
\usepackage[T1]{fontenc}
\usepackage{graphicx}

\usepackage{alltt}

\usepackage{amssymb}
\usepackage{amsmath}
\usepackage{theorem}

\theoremstyle{plain}
\newtheorem{cor}{Corollary}
\newtheorem{thm}{Theorem}
\newtheorem{lem}{Lema}

\theoremstyle{remark}
\newtheorem{Remark}{Remark}

\theoremstyle{definition}
\newtheorem{defi}{Definition}
\newtheorem{proof}{Proof}
\newtheorem{ex}{Example}

\newcommand{\N}{\mathbb{N}}

\title{Compound orbits break-up in constituents: an algorithm }

\author[1,2]{Jesús San Martín\thanks{jsm@dfmf.uned.es}}
\author[3]{A. González Gómez\thanks{antonia.gonzalez@upm.es}}
\author[1]{Mª José Moscoso \thanks{mariajose.moscoso@upm.es}}
\author[2]{Daniel Rodríguez-Pérez \thanks{daniel@dfmf.uned.es}}

\affil[1]{Departamento de Matemática Aplicada, E.U.I.T.I. Universidad
Politécnica de Madrid. 28012-Madrid, SPAIN}
\affil[2]{Departamento de Física Matemática y de Fluidos, Facultad
de Ciencias. Universidad Nacional de Educación a Distancia. 28040-Madrid,
SPAIN.}
\affil[3]{Departamento de Matemática Aplicada a los Recursos Naturales,
E.T. Superior de Ingenieros de Montes. Universidad Politécnica de
Madrid. 28040-Madrid, SPAIN.}

\begin{document}

\maketitle

\begin{abstract}
In this paper decomposition of periodic orbits in bifurcation diagrams are derived in unidimensional dynamics system $x_{n+1}=f(x_{n};r)$, being $f$ an unimodal function. We proof a theorem which states the necessary and sufficient conditions for the break-up of compound orbits in their simpler constituents. A corollary to this theorem provides an algorithm for the computation of those orbits. This process closes the theoretical framework initiated in (Physica D, \textbf{239}:1135--1146, 2010).
\vskip 2mm


\end{abstract}

\section{Introduction}

 Dynamical systems underlie in any Science we can imagine, from mathematical to social Sciences. Countless mathematical models have been developed to describe temporal evolution of the world around us: planets orbiting the Sun, flow of water in a river, people waving in a stadium, cells forming tissues in our body, cars moving along a road, etc. As a consequence of the extraordinary variety of phenomena studied, there exists a huge number of possible behaviors. One way to put so much complexity in order is by using symbolic dynamics \cite{Hao}.  In that case, the dynamical systems are modelled in a discrete space, resulting of breaking down the original dynamical system into finitely many pieces. Every piece is labelled by a symbol. System evolution is given by a sequence of symbols, each of them representing a piece of the system. The complexity is reduced because system evolution takes a finite number of states in contrast to infinity countless original states.  Although one might think  that no crucial information about the system may be obtained by this process there are some groundbreaking results in this subject. Special attention should be given to pioneering works by  Metropolis Stein \cite{metro} about symbolic sequences and by Milnor and Thurston \cite{Mil} who developed the kneading theory.
In particular, kneading theory is easier understood when the dynamical system
\begin{equation}\label{intro}
   x(n+1)=f(x(n))
\end{equation}
is ruled by an unimodal function, the function we will work with in this paper.

Some tools we will need to harness the powerful of symbolic dynamic are the periodic orbits. They have a periodic symbolic sequence and play an important role in dynamical systems, in particular the unstable ones as we will see later.

The composition law of Derrida, Gervois and Pomeau \cite{De}  allows to generate a symbolic sequence of complex structure from its constituents (periodic orbits). In particular, starting from the symbolic sequence of the supercycle of period one it is possible to build up symbolic sequences of Feigenbaum cascade orbits \cite{ref_1,ref_2}. So, one of the most important ways of transition to chaos is characterized. But not only that, by using saddle-node bifurcation cascades \cite{jes} and symbolic sequences of Feigenbaum cascade orbits, the structure of chaotic bands of the bifurcation diagram is also characterized (see figure \ref{figura1}).
Working with a unimodal function, the symbolic sequence is obtained as follows: the critical point of unimodal function is denoted by $\mathrm{C}$,  points located to the  right of  $\mathrm{C}$ are denoted $\mathrm{R}$ (right) and the ones located to its left as $\mathrm{L}$ (left). Obviously this process generates a problem. We can not distinguish an $\mathrm{R}$  $(\mathrm{L})$ from another, we do not know relative positions of points associated with them. Given the symbolic sequence  $\mathrm{CRLRRRLR}$, is the point of the first  $\mathrm{R}$ to right or to the left of the point associated with the second  $\mathrm{R}$? The flaw can be solved by associating orbits with permutations instead of with symbolic sequences, because the orbit points are labelled by numbers; those numbers give the relative point positions.  There are permutations giving the visiting order in Feigenbaum cascade orbits \cite{Mar} and there exists a composition law of permutations \cite{Mar1} replacing the composition law of Derrida, Gervois and Pomeau. Consequently, the characterization of bifurcation diagram structure is given by permutations and furthermore we have more specific information about point locations that provided by symbolic sequences.
We have just outlined how to build up the bifurcation diagram from its constituents.
From a mathematical point of view, however, it would be interesting to solve the inverse problem: what are the constituents of a complex structure?  More specifically, given a structure we would like to answer two questions:

\begin{enumerate}
\item[i)] Can we break down the structure? That is, is the structure made up of smaller constituents?
\item[ii)] If the answer to the first question is in the affirmative, how can we break down the structure and what are its constituents?
\end{enumerate}

Said otherwise, we are looking for necessary and sufficient conditions so that a structure can be decomposed into its constituents. That is the goal of this paper.
\begin{figure}
\begin{center}
\includegraphics[width=0.8\textwidth]{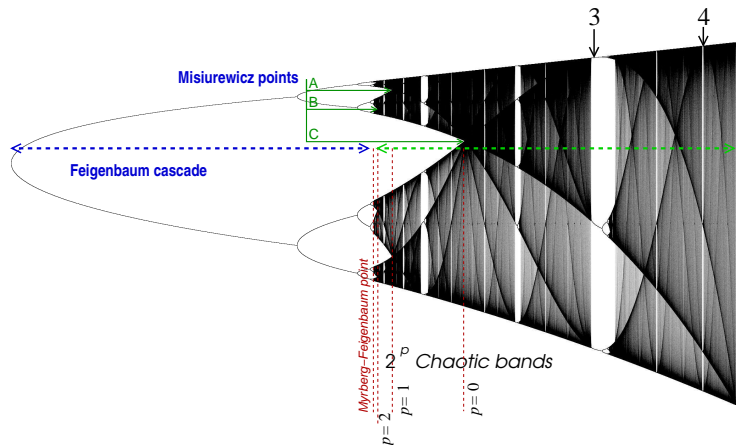}
\caption{\label{figura1} Canonical bifurcation diagram. $3$ and $4-$periodic windows are marked. Some Misiurewicz points (A, B, C) where chaotic bands merge are shown.}
\end{center}

\end{figure}

Solving the inverse problem of composition is already interesting because we complete the composition-decomposition problem. But the most important consequence is that decomposition process is not limited to stable orbits, unstable orbits will also be obtained in such a process. The knowledge of unstable orbits (limit cycles) is fundamental because they are the underlying skeleton of chaotic attractors \cite{Au,Cv}. The shorter the cycles, the better the approximation to the strange attractor \cite{Caroll}, that is why it is interesting to split large cycles into smaller constituents. On the other side, the unstable orbits in the skeleton are the corner-stones of  many chaos control techniques \cite{ott,py}. To implement theses techniques the unstable orbits need to be determined beforehand.

Orbits decomposition  can also be applied to continuous dynamical systems. They can be cast as discrete dynamical systems by using Poincaré section.  Points of Poincaré section corresponding to a continuous orbit lay out a periodic orbit in a discrete space.  If that orbit can be decomposed then the continuous orbit is a composed orbit.  Decomposition of theses orbits is crucial to calculate Gutzwiller trace formula \cite{gutz}, which relates spectra of quantum system with periodic orbits of the equivalent semiclassical system.  Roughly speaking, decomposition law of periodic trajectories will be useful every time cycle expansion techniques \cite{artu}  are used.

Decomposition law is also important from a practical or experimental point of view.  For example, if we have a $12-$periodic orbit we can be interested in knowing if the orbit is located in a primary period $12$ window  or in a period $3$ window  inside in a period $4$ window (see figure \ref{figura2}). 

\begin{figure}
\begin{center}
\includegraphics[width=0.55\textwidth]{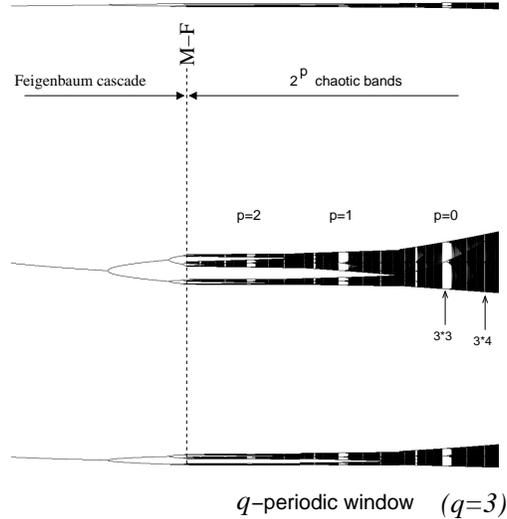}\\
\caption{ \label{figura2} Highlight of the ${3}-$periodic  window of figure \ref{figura1}.  This window mimics the canonical bifurcation diagram but repeated  three times. The ${3}\cdot 4-$periodic window is marked. }
\end{center}
\end{figure}

There exists a very intuitive way of looking at the decomposition process. If we have a period $hs$ orbit, that is with $hs$ points, we can imagine that every point is a chair in a room. The chairs are visited according to a permutation $\beta_s$. We split the $hs$ chairs into $h$ rooms with $s$ chairs each. We visit the rooms according with one permutation $\beta_h$ and every time we visit the same room we sit down in a different chair of the room according to another permutation $\beta_s$. We must find   $\beta_h$ and  $\beta_s$ from   $\beta_{hs}$. We are going to solve this task by using of couple of tricks. If we leave only one chair in every room the result would be like an $h$-periodic orbit such that a point is located at the critical point $\mathrm{C}$ of the unimodal function $f$ of (\ref{intro}) and the rest of points are located where $f$ is either increasing or decreasing. The chairs of a room located where $f$ is increasing (decreasing) are mapped into the next room preserving their relative location (flipped from right to left). So, we split the  $\beta_{hs}$ permutation into $h$ rooms of $s$ elements each, in such a way that images of these sets (except one of them) are either preserved or flipped from right to left. The set whose elements are neither preserved nor flipped will be  $\beta_s$, because they are the chairs of the room associated with the critical point  $\mathrm{C}$.

This papers is organized as follows. Firstly,  definitions and notations  are introduced. Second,   we prove decomposition theorem   to solve the mentioned problems. Then we develop an algorithmic to implement the theorem. We finish with our conclusions and discussions.  We will also  show some examples  to highlight how the  theorems and algorithms work.

\section{Definitions and notation}

Let $f:\mathrm{I}\to\mathrm{I}$ be an unimodal map with critical point at $\mathrm{C}$, that is $f$ is continuous  and strictly increasing (decreasing) on $[a,\mathrm{C})=\mathrm{J_{L}}$ and  strictly decreasing(increasing) on $(\mathrm{C}, b]=\mathrm{J_{R}}$. Without  loss of generality it can be assumed the critical point $\mathrm{C}$ is a maximum (see figure \ref{figura5}). So $f$ is decreasing in $\mathrm{J_{R}}$ and increasing in $\mathrm{J_{L}}$. Let  $O_{q}=\{x_{1}, \ldots, x_{q}\}=\{C, f(C), \ldots, f^{q}(C)\}$  be a $q-$periodic supercycle of $f$ and let  $\{C_{(1,q)}^{*},\, C_{(2,q)}^{*},\ldots,C_{(q,q)}^{*} \}$  be the set that denotes the descending cardinality ordering of the orbit $O_q$ \cite{Mar}. Let $f(C_{(i,q)}^*)$  be the next to $C_{(i,q)}^*$ (see \cite{Mar1}).

\begin{defi}\label{d1}   The natural number $\beta{(i,q)},\; i=1,...,q$ will denote the ordinal position of the cardinal point $f(C_{(i,q)}),\; i=1,...,q$ . That is $f(C_{(i,q)})= C_{(\beta{(i,q)},q)},\; i=1,...,q$ (see figure \ref{fighs}).  \end{defi}

\begin{Remark}If   $c$   denotes the ordinal position of the critical point $C$ of $f$, as  $f(C)$ is in the first  position  (see \cite{Mar} remark 1),  it results that  $\beta(c,q)=1$.
\end{Remark}

\begin{defi} We denote as $\beta_q$ the permutation $\beta_{q}=(\beta(1,q)\, \beta(2,q) \,\ldots  \beta(q,q))$. $\beta_q$ will be called the next visiting permutation of  $O_{q}$ (see  figure \ref{fighs}).
\end{defi}
\begin{Remark} If  the visiting order permutation is such that   $f(C_{(i,q)})=C_{(j,q)},$ that is  $C_{(i,q)}\rightarrow C_{(j,q)}$, we write
\[\left(
    \begin{array}{ccc}
      \cdots  & i & \cdots \\
      \cdots  & j & \cdots \\
         \end{array}
  \right)
\]
   then we reorder  the pairs $\left(\begin{array}{c} i \\ j \\ \end{array} \right)$ in such a way  that the index  $'i'$ has the natural order. For example, let   $O_4$ be a $4-$periodic orbit (see figure \ref{fighs})  with  visiting order  permutation \[   1\rightarrow \,  4\rightarrow \, 3\rightarrow \, 2 \] so we write
 \[
   \left(
     \begin{array}{cccc}
       1 & 4 & 3 & 2 \\
       \downarrow & \downarrow  & \downarrow & \downarrow \\
       4 & 3 & 2 & 1 \\
     \end{array}
   \right) \ \Longrightarrow \left(
     \begin{array}{cccc}
       1 & 2 & 3 & 4 \\
       \downarrow & \downarrow  & \downarrow & \downarrow \\
       4 & 1 & 2 & 3 \\
     \end{array}
   \right)\]
   after reordering we obtain the next  visiting permutations

   \[\beta_4=\left(
     \begin{array}{cccc}
              4 & 1 & 2 & 3 \\
     \end{array}
   \right)
   \]
\end{Remark}

\begin{figure}
\begin{center}
\includegraphics[width=0.86\textwidth]{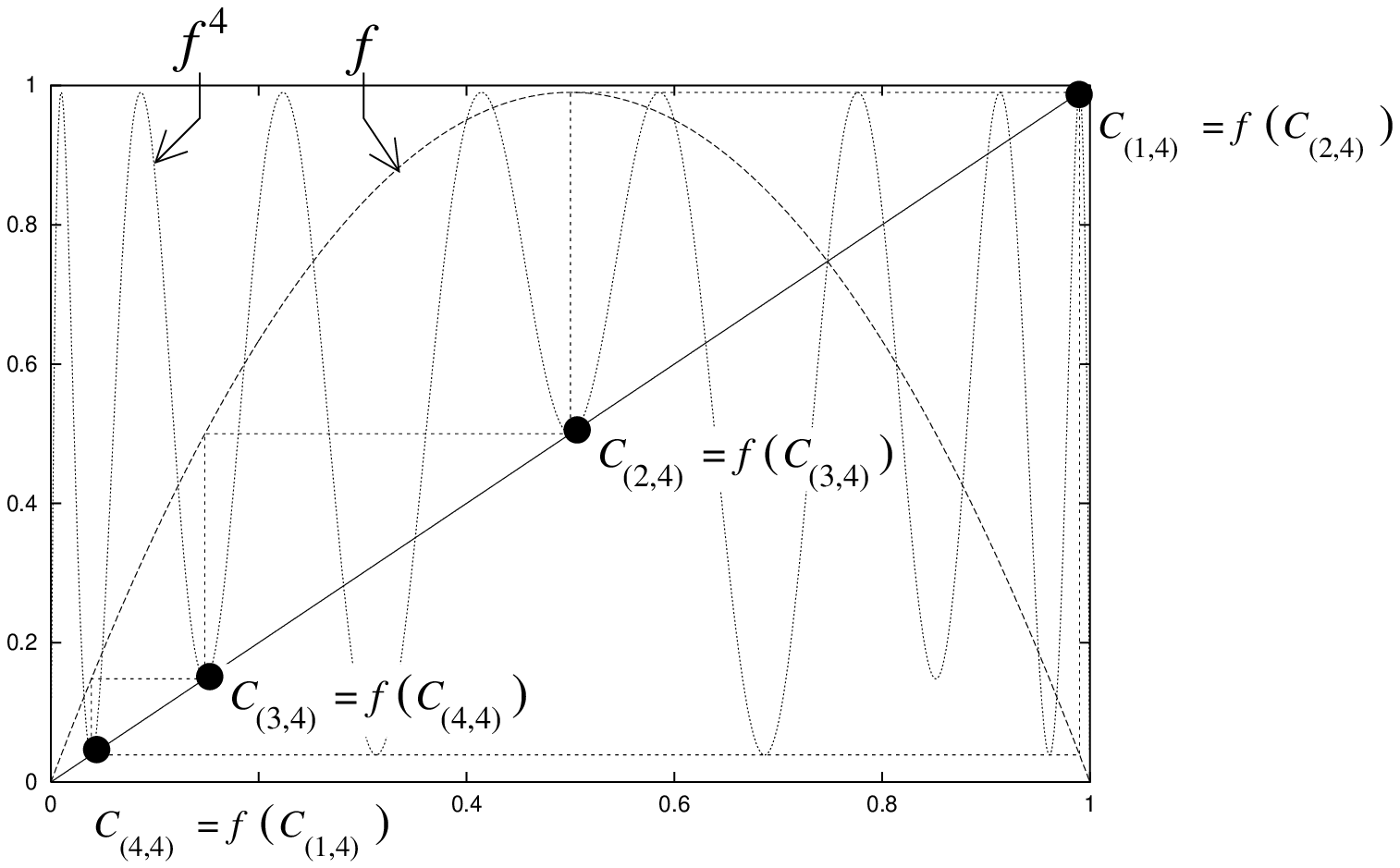}
\caption{\label{fighs} %
  \(\protect \begin{array}{cc} f(C_{(1,4)})=C_{(4,4)} &  f(C_{(2,4)})=C_{(1,4)} \protect \\ f(C_{(3,4)})=C_{(2,4)} &   f(C_{(4,4)})=C_{(3,4)}  \protect \end{array} \Longleftrightarrow \protect \protect\beta_{4}=\protect(4 \, 1 \, 2 \, 3 \protect)\)}
\end{center}
\end{figure}

\begin{defi}\label{d32} Let  $\beta_q$  be the next visiting permutation of  $O_q$ and let $q=hs$. We define the   $j-box$   of   $O_q$ by
$H_j=\{(j-1)s+k; \, k=1,\ldots, s\}$ for $j=1, \ldots, h$. We denote by $\beta_q(H_j)$ the set given by
\[\beta_q(H_j)=\{\beta((j-1)s+k,q): \,  k=1,\ldots, s \}\]
\end{defi}

\begin{defi}\label{d3} Let  $\beta_q$  be the next visiting permutation of  $O_q$ and let $q=hs$. We denote  by  $(\beta^j_q)$   with  $j=1,\ldots, h$
\[\beta^j_q=\left(\begin{array}{ccc} (j-1)s+1 & \, \ldots  \, & (j-1)s+s\\ \beta((j-1)s+1,q) & \ldots  \,&  \beta((j-1)s+s,q)  \\                                                                                                                            \end{array}   \right)\]
and  $\beta^j(r,q)=\beta((j-1)s+r,q)$  with $r=1, \ldots, s$ (see figure \ref{figura5}).

\end{defi}

\begin{figure}
\begin{center}
\includegraphics[width=0.6\textwidth]{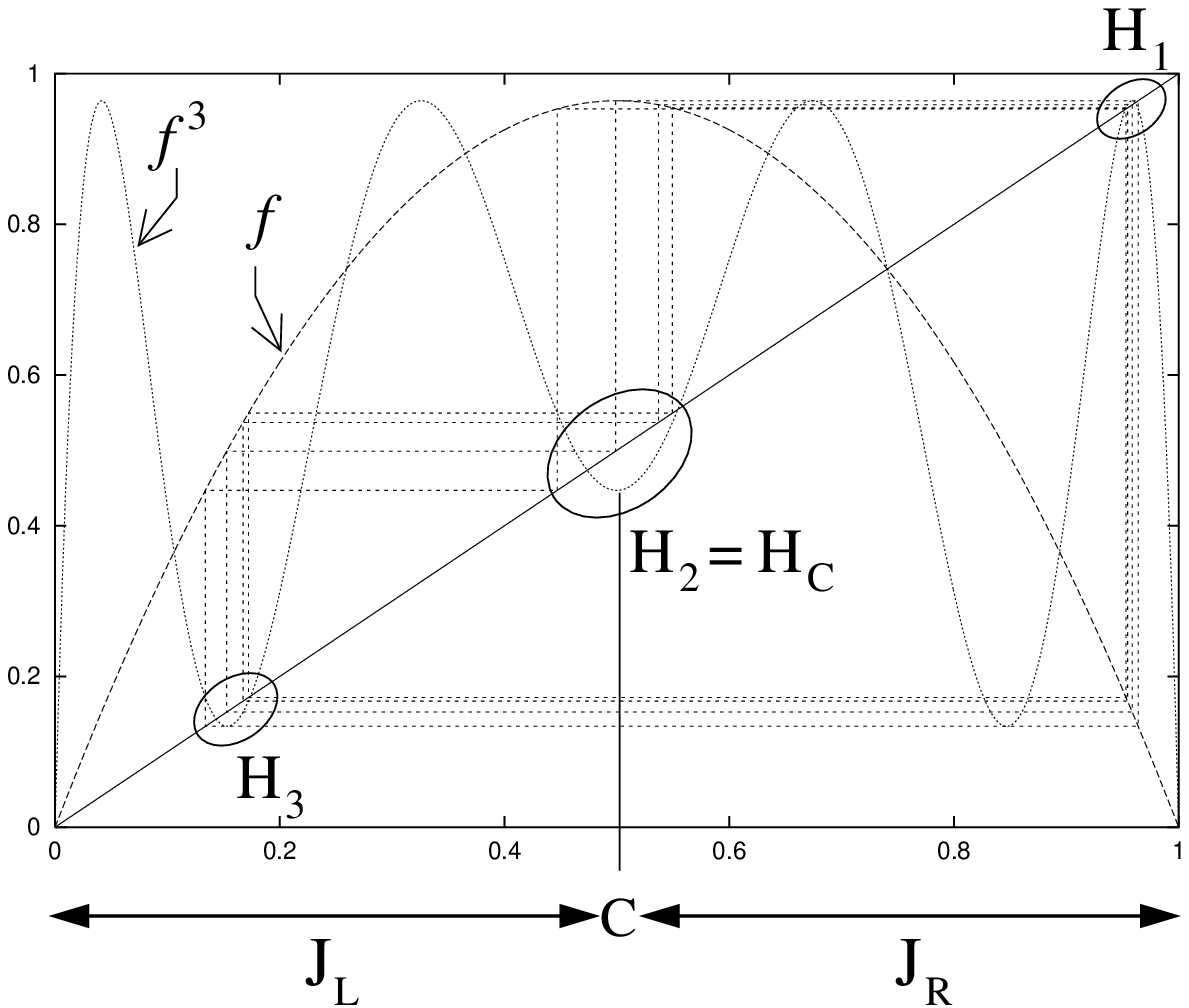}\\
\caption{ \label{figura5} \(\protect  \beta_{12} = \left( \protect\begin{array}{cccccccccccc}12 & 11 & 10 & 9 & 3 & 2&1&4&5&6&7&8  \protect\end{array}  \right) \) is the next visiting permutation of $O_{12}$. If we consider $h=3$ and $s=4$
then \protect\\ \(\protect \beta_{12}^2=\left( \protect\begin{array}{cccc} 5 & 6 & 7&8  \protect\\   3 & 2& 1 & 4\protect\end{array} \right) \protect\)\protect\begin{tabular}{c}
 $\longleftrightarrow \protect \quad  \mathrm{H}_{2}$ \protect\\
$\longleftrightarrow  \protect \quad \mathrm{H}_{1}$
\protect\end{tabular}}
\end{center}
\end{figure}

\begin{defi}\label{def2} Let   $\gamma_{n}$ be a permutation of $n$ elements.  We define the inversion permutation of $\gamma_{n}$, denoted by    $\gamma^*_{n}$, as the permutation given by $\gamma^*_{n}=(\gamma^*(1,n),\ldots, \gamma^*(n,n))$ with
\[
 \gamma^*(i,n)=n+1-\gamma(i,n) \qquad  i=1,\ldots,n
\]

Notice that if  $I_n=(I(1,n),\ldots, I(n,n))$ is  the identity permutation then $I^*_n \circ I^*_n= I_n$.
\end{defi}

\begin{defi}\label{def3}  Let   $\gamma_{n}$ be a permutation of $n$ elements. We define the conjugated permutation of  $\gamma_{n}$, denoted by      $\overline {\gamma}_{n}$,  as the permutation given by $\overline {\gamma}_{n}=(\overline{\gamma}(1,n), \ldots, \overline{\gamma}(n,n))$ with
\[
  \overline{\gamma}(i,n)=n+1- \gamma(n+1-i,n) \qquad  i=1,\ldots,n
\]
\end{defi}

\section{Theorem of periodic orbit decomposition}

 In order to obtain the decomposition theorem below, we need to revisit the composition process and state it in terms of next visiting permutations.

 Let $O_h, O_s$ be  supercycles of a $C^2$-unimodal map $f$ with next visiting permutations $\beta_{h}$ and $\beta_{s}$ respectively. The geometric meaning of composing $O_h$ and $O_s$  involves replacing the  $h$ points of $O_h$ by $h$ boxes, with $s$ points each,  such that all points of a same box  are mapped into the same box.
It is important to point out that boxes are visited  consecutively according to $\beta_h$  and that every time the same box is visited then the box points are visited according to $\beta_s$, if $f^h$ has a maximum and according to $\overline{\beta}_s$ if  $f^h$ has a minimum (see \cite{Mar1} for more details). As boxes (see  definition \ref{d32}) $H_i$ with $i=1,\dots,h$, are visited according to $\beta_h$ we split the visit in two parts:

  \begin{equation}\label{s1}
  H_c\longrightarrow H_1
   \end{equation}

\begin{equation}\label{s2} H_1\longrightarrow\ldots H_i\ldots \longrightarrow H_c
  \end{equation}

In sequence $(\ref{s2})$, excluding $H_c$,  boxes are located in $\mathrm{J_R}$ or  $\mathrm{J_L}$. Every time the orbit leaves a box located  in  $\mathrm{J_L}$ the  points in that box are mapped according to  the identity permutation, $I_s$, because $f$ is increasing in  $\mathrm{J_L}$. On the contrary, every time the orbit leaves a box located  in  $\mathrm{J_R}$, the  points in that box are mapped reverted from left to right because $f$ is decreasing in  $\mathrm{J_R}$, that is, they are linked by $I^*_s$.  As $I^*_s\circ I^*_s=I_s$ it results  that $H_1$ is linked with $H_c$ by $I_s$ or $I^*_s$.
It only remains to know the link between $H_c$ and $H_1$ of sequence $(\ref{s1})$ to close the orbit  (see figure \ref{figura5})

This is the geometrical mechanism underlying the proof of the following lemma. This lemma is essential to prove the  theorem \ref{T2} which is the goal of this paper.

\begin{lem}\label{L1}  Let $O_h$, $O_s$ be two supercycles of a $C^2$-unimodal map $f$ with next visiting permutations $\beta_{h}$ and $\beta_{s}$ respectively. Let $c$ be such that $\beta(c,h)=1$. If $O_{hs}$ is the supercycle resulting of composing $O_h$ with $O_s$ then its next visiting  permutation
$$\beta_{hs}=(\beta^1(1,hs), \ldots, \beta^1(s,hs), \beta^2(1,hs), \ldots,  \beta^2(s,hs), \ldots, \beta^h(1,hs), \ldots, \beta^h(s,hs))$$
is given, for all $k=1, \ldots,s$,  by
\begin{itemize}
  \item [a)]\, if $i=1,\ldots, c-1$
\[\beta^i{(k,hs)}=  \beta(i,h)s -(k-1)\]
 \item [b)]\, if $i=c$
 \[ \beta^c {(k,hs)}=   \left\{                   \begin{array}{ll}
                     \beta{(k,s)}  & \hbox{if $i=c$ is odd} \\
                      \beta {(s+1-k,s)} & \hbox{if $i=c$ is even}
                    \end{array}\right.
\]      \item [c)]\,  if  $i=c+1,\ldots,h$
       \[\beta^i{(k,hs)}= (\beta(i,h)-1)s +k \]
       \end{itemize}
\end{lem}

\begin{proof}
	
As $i$-th box is  preceded by  $(i-1)$ boxes with $s$ elements  each, the elements of  $i$-th box are given by
$(i-1)s+k,\;  k=1, \ldots,s$.

As $i$-th  box is mapped into  $\beta(i,h)$-th box, it results that the $s$ elements of the $i$-th  box are mapped into   the $s$ elements of the $\beta(i,h)$-th box. In order to know the images of the $s$ elements in the $i$-th  box we have to consider where the  $i$-th  box is located.

\begin{enumerate}

\item [a)] \, 	 $i$-th  box  located in   $\mathrm{J_R}$, that is,  $i=1,\ldots,c-1$.

As $f$ is strictly decreasing in  $\mathrm{J_R}$ the order of the elements in $i-$th  box, are reverted from left to right after mapping into $\beta(i,h)-$th box, that is,

$$(i-1)s+k  \longrightarrow \beta(i,h) s - (k-1)  \ \text{ with } \, k=1,\ldots, s$$

so $ \beta^i{(k,hs)}= \beta(i,h)s -(k-1) \ \text{ with } \, k=1,\ldots, s$ if  $i=1,\ldots,c-1$.

  \item [b)]	\,  $i-$th  box  located in  $\mathrm{J_L}$, that is $i=c+1,\ldots,h$.

As $f$ is strictly increasing in  $\mathrm{J_L}$ the elements of $i-$th  box, are mapped into the elements of   $\beta(i,h)-$th box conserving  their relative order, that is

$$(i-1)s+k  \longrightarrow (\beta(i,h) - 1)s+k  \ \text{ with } \quad  k=1,\ldots, s$$

so $\beta^i{(k,hs)}= (\beta(i,h)-1)s +k  \ \text{ with } \, k=1,\ldots, s$ if $i=c+1,\ldots,h.$

	 \item [c)]\,  The  $i-$th  box  is $H_c$, the so-called central box. The proof splits  into two steps:

\begin{itemize}
  \item [c.1)]   $c$ is odd. As $C$ is odd the number of points of $O_h$ located in $\mathrm{J_R}$ is even (in  \cite{Mar1} this is said as the $\mathrm{R}$-parity of $\mathrm{I_1}, \ldots, \mathrm{I_{h-1}}$ is even (definition $2$  \cite{Mar1}), so  $f^h$ has a maximum (lemma $3$ \cite{Mar1}) and  then the link of a point of the central box with the next visiting point in this same box is given by $\beta_{s}$, as we have just explained above. But the linking of these two points requires visiting all boxes before they connect between themselves. Therefore   as the number of  $O_h$  located in $\mathrm{J_R}$ is even if we set off $H_{1}$ to reach $H_{c}$ we will have visited an even number of boxes located in $\mathrm{J_R}$. Given that images of points located in $\mathrm{J_R}$, where $f$ is decreasing, are reverted from left to right and two reversion are equivalent to an identity, it results that the $s$ elements of $H_{1}$ are linked with the $s$ elements of $H_{c}$ by the identity permutation $I_s$.  So, we have to connect the central box with the first one by  an unknown permutation, $\gamma_{s}$, such that $I_s\circ \gamma_s=\beta_s$. Then  $\gamma_s=\beta_s$.

 So  the elements  of the central box, given by $(c-1)s+k  \qquad   k=1, \ldots,s$, are mapped into the elements of first   box by
 $\beta_s$,  that is,

$$(c-1)s+k  \longrightarrow \beta(k,s) \ \text{ with } \, k=1,\ldots, s$$
 so $\beta^c{(k,hs)}= \beta{(k,s)} \ \text{ with } \, k=1,\ldots, s$ if $c$ is odd.

   \item [c.2)]  $c$ is even. By a similar argument to the one given above, the elements of $H_{1}$ and $H_{c}$ are linked by $I^*_{s}$ given that there is an odd number of reversions. Furthermore, the link of a point of the central box with the next visiting point in this same box is given by $\overline{\beta}_s$ because $f^h$ has a  minimum \cite{Mar1},  as we have just explained above. So, we have to connect the central box with the first one by  an unknown  permutation, $\gamma_{s}$, such that, $I^*_{s}\circ \gamma_s=\overline{\beta}_s$. Then $I^*_{s}\circ I^*_{s}\circ \gamma_s=I^*_{s}\circ \overline{\beta}_s.$ Since $\overline{\beta_s}=I^*_s \circ \beta_s \circ I^*_s$  we have $ \gamma_s=\beta_{s}\circ I^*_{s}.$

       So $\beta^c{(k,hs)}= \beta (s+1-k,s) \ \text{ with } \, k=1,\ldots, s$ if $c$ is even.

\end{itemize}

  \end{enumerate}
 \end{proof}

\begin{Remark}\label{R1} Notice that, under conditions of lemma \ref{L1}, if $O_{hs}$ is the composed supercycle of $O_h$ with $O_s$, when $\beta(c,h)=1$ with $c$ even, its next visiting  permutation $\beta_{hs}=\left(  \beta(j,hs)   \right)$  is given by
\[\left(
  \begin{array}{c|c|c}\underbrace{\begin{array}{c}    (i-1)s+k\\ \\ \beta(i,h)s-(k-1)  \end{array}}_{\substack{k=1, \ldots, s\; \forall i \\ i=1,\ldots,c-1 }} & \underbrace{\begin{array}{ccc} (c-1)s+k \\ \\ \beta(s+1-k,s) \end{array}}_{\substack{k=1, \ldots, s \\ i=c}} &  \underbrace{\begin{array}{c} (i-1)s+k\\  \\(\beta(i,h)-1)s+k\end{array}}_{\substack{k=1, \ldots, s\; \forall i \\ i=c+1,\ldots,h}}\end{array}
\right)\]
Notice also, that if $i<c$, then $\beta^i(r+1,{hs})=\beta^i(r,{hs})-1$ for all  $r=1, \ldots, s-1$, whereas if $i>c$ then $\beta^i(r+1,{hs})=\beta^i(r,{hs})+1 $ for all  $r=1, \ldots, s-1$, and that $\{\beta^c(r,{hs})\}_{r=1, \ldots, s}\equiv \{1,\ldots,s\}$.
\end{Remark}

Our next step is to determine necessary and sufficient conditions in order to know whether a periodic orbit is compound or not. Below, an algorithm will be given to break-up periodic orbits into their constituent elements.

\begin{Remark}     $\left[ \; \cdot \; \right]$ means   integer part of a real number.
\end{Remark}

\begin{thm}\label{T2}Let   $O_{q}$  be a supercycle of a   $C^2-$unimodal map $f$ with   the  next visiting  permutation $\beta_{q}=(\beta(1,q), \ldots, \beta(q,q))$, and $\beta(z,q)=1.$  Let $h,s\in \N$ be  such that $q=hs$.
 $O_q$ is the  composition of two supercycles $\text{O}_{h}$ and $\text{O}_{s}$ if only if $\beta_q$ is given,  for all $k=1, \ldots, s$, by
\begin{itemize}
 \item [a)] \,  if \, $i=1, \ldots, [\frac{z}{s}]$
  \[ \beta^i(k,q)=
  \beta^i(1,q)-(k-1)
\]
  \item [b)] \,  if \, $i=[\frac{z}{s}]+2, \ldots, h$
  \[ \beta^i(k,q)=
  \beta^i(1,q)+(k-1)
\]
  \item [c)] \,  if \, $i=[\frac{z}{s}]+1$
\[   \beta^i{(k,q)}=   \left\{                   \begin{array}{ll}
                     \beta{(k,s)}  & \hbox{if $i=[\frac{z}{s}]+1$ is odd} \\
                      \beta {(s+1-k,s)} & \hbox{if $i=[\frac{z}{s}]+1$ is even}
                    \end{array}\right.
\]
 \end{itemize}
where $\beta (k,s)$ is the $k-$th element of a next visiting  permutation, $\beta_s,$  of an orbit with period $s$.
\end{thm}

 \begin{proof} ~
 
 $\Longrightarrow\,)$ Let  $O_q$ be the composition of two supercycles $\text{O}_{h}$ and $\text{O}_{s}$. Let $\beta_h$ and $\beta_s$ be the next visiting permutations of $O_h$ and $O_s$, respectively. As $\beta(z,q)=1$ then $\beta(\left[\frac{z}{s}\right]+1, h)=1$. If $i \neq \left[\frac{z}{s}\right]+1$, by lemma \ref{L1}, we have
  \begin{equation}\label{te2}
   \beta^i(k,q)=\left\{
    \begin{array}{ll}
      \beta(i,h)s -(k-1) & \hbox{if $i=1,\ldots, [\frac{z}{s}]$} \\
      (\beta(i,h)-1)s +k & \hbox{if  $i=[\frac{z}{s}]+2,\ldots,h$}
    \end{array}
  \right.
  \end{equation}
  It follows from (\ref{te2})
     \begin{equation}\label{te3}\beta^i(1,q)=\left\{ \begin{array}{ll}
      \beta(i,h)s  & \hbox{if $i=1,\ldots, [\frac{z}{s}]$} \\
      (\beta(i,h)-1)s +1 & \hbox{if  $i=[\frac{z}{s}]+2,\ldots,h$}
    \end{array}
    \right.  \end{equation}
after substituting (\ref{te3}) in equation (\ref{te2}), we get for $i\neq [\frac{z}{s}]+1$:
\begin{equation}\label{te4}
    \beta^i(k,q)=\left\{ \begin{array}{ll}
      \beta^i(1,q) -(k-1) & \hbox{if $i=1,\ldots, [\frac{z}{s}]$} \\
      \beta^i(1,q)+ (k-1)& \hbox{if  $i=[\frac{z}{s}]+2,\ldots,h$}
    \end{array}\right.
\end{equation}
   The case $i= [\frac{z}{s}]+1$, follows directly from b) in lemma \ref{L1}.

 $\Longleftarrow\, )$ We assume that $\beta_q$ satisfies conditions a), b), c) of theorem \ref{T2} and want to proof that $O_q$ is the composition of two supercycles $O_h$ and $O_s$. For this we will build up two next visiting  permutations $\beta_s$ and $\beta_h$ whose composition is $\beta_q$.

 We define  $\beta_{s}=(\beta(1,s), \ldots, \beta(s,s))$ where

      \begin{equation}\label{te6} \beta(k,s)=\left\{
  \begin{array}{ll}
  \beta^{[\frac{z}{s}]+1}(k,q)& \hbox{if \, $[\frac{z}{s}]+1$ is odd} \\
&\\
  \beta^{[\frac{z}{s}]+1}(s+1-k,q) & \hbox{if  \, $[\frac{z}{s}]+1$ is even}
  \end{array}
\right.\end{equation} As $\beta_q$ verifies condition $c)$ in theorem \ref{T2} it results that (\ref{te6}) is  a  next visiting permutations of a $s-$periodic orbit $O_s$.

 Now we define $\beta_{h}=(\beta(1,h), \ldots, \beta(h,h))$ with
      \begin{equation}\label{te7} \beta(i,h)=\left\{
  \begin{array}{cl}
  \dfrac{\beta((i-1)s+1,q)}{s} & \hbox{ $i=1, \ldots, [\frac{z}{s}]$} \\
&\\ 1  & \hbox{ $i= [\frac{z}{s}]+1$} \\ & \\
  \dfrac{\beta((i-1)s+1,q)+(s-1) }{s} & \hbox{$i=[\frac{z}{s}]+2, \ldots, h$}
  \end{array}
\right.\end{equation}
In order to prove that $\beta_h$ is a next visiting permutation, one of the
things we have to prove is that the set $\{ \beta(i,h): \, i=1, \ldots, h\}$
coincides with the set $\{1, \ldots, h\}$. Let us study the different values of
$i$ in (\ref{te7}).
 \begin{itemize}
 \item Let $i=1, \ldots, [\frac{z}{s}]$. According to (\ref{te7}), it results
 \begin{equation}\label{te8} \beta(i,h)=\dfrac{\beta((i-1)s+1,q)}{s}\end{equation}
  Given that for every $i=1, \dots, h$ there exists only one $j \in \{1, \dots, h\}$
  such that $\beta_q(H_i)=H_j$ (see appendix), it results that
  \begin{equation}\label{te7a}
   \beta((i-1)s+1,q)=(j-1)s+r,\, r=1,\ldots, s
  \end{equation}
  Taking into account (\ref{te8}) and (\ref{te7a}), in order to prove that
  $\beta(i,h)$ is a natural number let us see that $\beta((i-1)s+1,q)=(j-1)s+s$.
  Let us assume it were false, that is,
   \begin{equation}\label{te9}\beta((i-1)s+1,q)=(j-1)s+r \ \text{ for some } \ r=1,\ldots, s-1\end{equation}
   Applying condition $a)$ for $k=s$, and taking into account definition \ref{d3}, it yields
 \begin{equation}\label{te10}\beta((i-1)s+s,q)=\beta((i-1)s+1,q)-(s-1).\end{equation}
  Then from equation (\ref{te9}) and  equation (\ref{te10}), it results
 \begin{equation}\label{te10a}\beta((i-1)s+s,q)=(j-1)s+(r+1-s) \ \text{ for some } r=1,\ldots, s-1\end{equation} from (\ref{te10a}), given that  $r+1-s\leq 0$,  $\beta((i-1)s+s,q)\notin H_j$ which is in contradiction with $\beta_q(H_i)=H_j$ (see appendix). So  $\beta((i-1)s+1,q)=(j-1)s+s$  and replacing it in (\ref{te8}), we obtain
   \begin{equation}\label{te7b}
   \beta(i,h)=\dfrac{\beta((i-1)s+1,q)}{s}=j\in \{1, \ldots, h\}
   \end{equation}
  According to c) of this theorem, it holds $\beta_q(H_{\left[\frac{z}{s}\right]+1})=H_1$. Given that $i \le \left[\frac{z}{s}\right]$ it results that $j \neq 1$ in (\ref{te7b}). Hence, $j\in \{2, \ldots, h\}$.

  \item Let $i=\left[\frac{z}{s}\right]+2, \ldots, h$. Taking into account definition \ref{d3} and condition b) of this theorem, it results from (\ref{te7}) that
   \begin{equation}\label{te10b}\ \beta(i,h)=\dfrac{\beta((i-1)s+s,q)}{s}\end{equation}
  As  $ \beta((i-1)s+s,q)=(j-1)s+s$ (proof is similar to the case $i\leq [\dfrac{z}{s}]$) it results, from (\ref{te10b}), that for every $i\geq [\dfrac{z}{s}]+2$ there exists only one $j\in \{2, \ldots, h\}$ such that $\beta(i,h)=j$. Furthermore, these  $j\in \{2, \ldots, h\}$ are different from those obtained for the case $i\leq \left[\frac{z}{s}\right]$ (because for every $i=1,\dots,h$ there exists only one $j\in \{1, \ldots, h\}$, such that $\beta_q(H_i)=H_j$, see appendix).
  
  \item Let $i=\left[\frac{z}{s}\right]+1$. According to c) of this theorem, $\beta_q(H_{\left[\frac{z}{s}\right]+1})=H_1$, that is, $j=1$.

\end{itemize}

Consequently, the set $\{\beta(i,h):  i=1, \ldots, h\}$ coincides with the set $\{1, \ldots, h\}$.

Our final goal is to prove that $O_q$ is the composition of $O_h$ and $O_s$, that is, $O_q\equiv O_{hs}$.

We denote by  $O_h$  the $h-$periodic orbit whose next visiting permutation is given by  $\beta_h$ (see eq. \ref{te7}). We denote by $O_s$ the orbit of period $s$, whose next visiting permutation is given by $\beta_s$ (see eq. \ref{te6}).

 According to lemma  \ref{L1} for $i\neq [\frac{z}{s}]+1$ it holds
   \begin{equation}\label{e8} \beta^i{(k,hs)}=\left\{
    \begin{array}{ll}
      \beta(i,h)s -(k-1) & \hbox{$i=1,\ldots, [\frac{z}{s}]$} \\
      (\beta(i,h)-1)s +k & \hbox{ $i=[\frac{z}{s}]+2,\ldots,h$}
    \end{array}
  \right.
  \end{equation}
  by taking account (\ref{te7}), (\ref{e8}) is rewritten  as
 \begin{equation} \beta^i{(k,hs)}=\left\{
  \begin{array}{ll}
  \beta((i-1)s+1,q)-(k-1) & \hbox{ $i=1, \ldots, [\frac{z}{s}]$} \\
&\\
  \beta((i-1)s+1,q)+ (k-1)   & \hbox{$i=[\frac{z}{s}]+2, \ldots, h$}
  \end{array}
\right.
\end{equation}

According to lemma  \ref{L1} for $i= [\frac{z}{s}]+1$ it holds
 \begin{equation}\label{e9}
 \beta^{[\frac{z}{s}]+1} {(k,hs)} =  \left\{
                    \begin{array}{ll}
                     \beta(k,s) & \hbox{ $i=[\frac{z}{s}]+1$ is odd} \\
                      \beta(s+1-k,s) & \hbox{if $i=[\frac{z}{s}]+1$ is even}
                    \end{array}
                  \right.
                  \end{equation}
         by using (\ref{te6}), (\ref{e9}) is rewritten  as
            \begin{equation}  \beta^{[\frac{z}{s}]+1}{(k,hs)} =\left\{
                    \begin{array}{ll}
                     \beta^{[\frac{z}{s}]+1}{(k,q)}  & \hbox{ $i=[\frac{z}{s}]+1$ is odd} \\
                    \beta^{[\frac{z}{s}]} {(s+1-(s+1-k),q)} & \hbox{ $i=[\frac{z}{s}]+1$ is even}
                    \end{array}
                  \right\}
\end{equation}
so $ \beta^{[\frac{z}{s}]+1}{(k,hs)}=\beta^{[\frac{z}{s}]+1}(k,q).$

By hypothesis of the theorem both $O_{hs}$ and $O_s$ are admissible orbits, it remains to be seen that $O_h$ is also an admissible one.  By construction the $h$ first elements  of the symbolic sequence of  $O_{hs}$ coincide with the symbolic sequence of  $O_{h}$, therefore by using shift operator and the kneading theory if  $O_{h}$ were not an admissible orbit neither  $O_{hs}$ would be  [1,2], that is a contradiction, consequently  $O_{h}$ is an admissible orbit.

Therefore $\beta_{hs}=\beta_q$. As $\beta_s$ and $\beta_h$ are the next visiting permutations of $O_s$ and $O_h$ respectively it yields   that $O_q$ is the composition of $O_h$ and $O_s$.

\end{proof}

 Notice that theorem \ref{T2} not only provides the decomposition of a compound periodic orbit, but also it allows to deduce what orbits are not decomposable because they are not the composition of two orbits (see example \ref{ejem1}, below).
 Given that a compound orbit is associated with a window inside a window (in the bifurcation diagram), a non decomposable orbit is associated with a primary window. This is a direct application of the theorem to distinguish primary windows from windows inside windows.

\begin{ex}\label{ejem1} Let the visiting sequence of the ${15}-$periodic orbit be  given by
\[                            1\rightarrow \,  15\rightarrow \, 8\rightarrow \,  7 \rightarrow \,  9\rightarrow \,    6\rightarrow \,  10\rightarrow \,  5 \rightarrow \,  11\rightarrow \,  4\rightarrow \,   12\rightarrow \,  3 \rightarrow \,   13\rightarrow \,  2\rightarrow \,  14
                         \] so its next visiting permutations is
\begin{equation}\label{ejem-no-descomp}
\beta_{15}=(\begin{array}{ccccccccccccccc}
                            15& 14&  13&  12& 11&  10&  9&7 &6 & 5& 4 & 3 &  2&  1&  8
                           \end{array})
\end{equation}
 By using theorem \ref{T2} we are going to prove that the orbit (\ref{ejem-no-descomp}) is not the composition of an orbit $O_{h=3}$ and an orbit $O_{s=5}$. As $[\frac{z}{s}]+1=[\frac{14}{5}]+1=3$ is odd, if $O_{15}$ were a compound orbit, it would follow from theorem \ref{T2} that $\beta^3(k,15)$ is the $k$-th element of a next visiting permutation $\beta_5$ corresponding to the orbit $O_5$, $k=1,\,2,\,3,\,4,\,5$. However, according to (\ref{ejem-no-descomp}) we have 
 \[ \beta_5=(
 \begin{array}{ccccc}
  \beta^3(1,15)=4 & \beta^3(2,15)=3& \beta^3(3,15)=2 & \beta^3(4,15)=1 & \beta^3(5,15)=8
 \end{array})\]
 that is, $\beta_5=(4 \, 3 \, 2 \, 1 \, 8)$ which is not a period $5$ orbit (because it has an $8$ in it). In the same way, it is proven that it is not the composition of $O_{h=5}$ and an orbit $O_{s=3}$.

\end{ex}
\section{Algorithm}

The following corollary to theorem \ref{T2} provides the decomposition algorithm:

\begin{cor}\label{cor1}
 Let $O_q$  be a supercycle of a   $C^2-$unimodal map $f$ with   the  next visiting  permutation $\beta_q$.  Let  $z$  be such  that $\beta(z,q)=1.$ If $O_q$ is the result of composing two supercycles $O_{h}$ and $O_{s}$ then the next visiting permutations $\beta_h$ and $\beta_s$ are giving by
 \begin{equation*} \beta(k,s)=\left\{
  \begin{array}{ll}
  \beta^{[\frac{z}{s}]+1}(k,q)& \hbox{if \, $[\frac{z}{s}]+1$ is odd} \\
&\\
  \beta^{[\frac{z}{s}]+1}(s+1-k,q) & \hbox{if  \, $[\frac{z}{s}]+1$ is even}
  \end{array}
\right.\end{equation*}
and
 \begin{equation*} \beta(i,h)=\left\{
  \begin{array}{ll}
  \dfrac{\beta((i-1)s+1,q)}{s} & \hbox{ $i=1, \ldots, [\frac{z}{s}]$} \\
&\\ 1  & \hbox{ $i= [\frac{z}{s}]+1$} \\ & \\
  \dfrac{\beta((i-1)s+1,q)+ (s-1)  }{s} & \hbox{$i=[\frac{z}{s}]+2, \ldots, h$}
  \end{array}
\right.\end{equation*}

\end{cor}

Contrary to example \ref{ejem1}, the following example will apply theorem \ref{T2} to a decomposable orbit. Notice that the factorization of a natural number is not unique. For instance, a compound $12-$ periodic orbit could be associated with: a $3-$periodic window inside a $4-$periodic one, a $4-$periodic inside a $3-$periodic, a $2-$periodic inside a $6-$periodic,  or a $6-$periodic inside a $2-$periodic one. The  theorem \ref{T2} gives the only admissible decomposition. This theorem will allow us to know what window the orbit is effectively located in. Later, by using the algorithm, we will obtain from what next visiting  permutations the orbit has been composed.

\begin{ex}\label{ejem2}Let  the   visiting sequence of the ${12}-$periodic orbit be  given by
\[                            1\rightarrow \,  12\rightarrow \, 8\rightarrow \,  4 \rightarrow \,  9\rightarrow \,    5 \rightarrow \,  3\rightarrow \,  10\rightarrow \,   6\rightarrow \,  2 \rightarrow \,   11\rightarrow \,  7
                         \]
so its next visiting permutation of $O_{12}$ is
\begin{equation}\label{beta12}
\beta_{12} = \left(
              \begin{array}{cccccccccccc}
                               12 & 11 & 10 & 9 & 3 & 2&1&4&5&6&7&8 \\
              \end{array}
            \right)
\end{equation}
According to this, we have that $z=7$. So  $\beta_{12}$ could be decomposed as $\beta_{2}\circ \beta_{6}$, $\beta_{6}\circ \beta_{2}$, $\beta_{4}\circ \beta_{3}$ or  $\beta_{3}\circ \beta_{4}$.

Firstly, we will use necessary conditions of theorem \ref{T2} to reject the cases $\beta_{2}\circ \beta_{6}$, $\beta_{6}\circ \beta_{2}$ and $\beta_{4}\circ \beta_{3}$ (items 1, 2 and 3 below). Then, we will see that $\beta_{3}\circ \beta_{4}$ satisfies the sufficient conditions of the theorem (item 4 below).

\begin{enumerate}
\item If $h=2$ and $s=6$, then $\left[\frac{z}{s}\right]=\left[\frac{7}{6}\right]=1$.

According to a) of theorem \ref{T2}
\begin{equation}\label{eq15-1}
 \beta^1 (5,12) = \beta^1 (1,12) - 4
\end{equation}
From (\ref{beta12}) it holds that
\[
 \beta^1 (1,12) = 12
\]
consequently (\ref{eq15-1}) is rewritten as
\[
 \beta^1 (5,12) = 8
\]
However, according to (\ref{beta12})
\[
 \beta^1 (5,12) = 3
\]
what is a contradiction.

\item If $h=6$ and $s=2$, then $\left[\frac{z}{s}\right] + 1=\left[\frac{7}{2}\right]+1=4$, which is even.
According to c) of theorem \ref{T2}, we have that $\beta^4 (1,12)$ and $\beta^4 (2,12)$
should be the elements of a next visiting permutation $\beta_s = \beta_2$. However, according to
(\ref{beta12}), $\beta^4 (2,12)=4$ which is a contradiction.

\item If $h=4$ and $s=3$, then it is solved analogously to case 1.

\item In order to see that $\beta_3 \circ \beta_4$ satisfies the sufficient conditions of theorem \ref{T2}
let us see that it satisfies conditions a), b) and c) of theorem \ref{T2}.
\begin{enumerate}
 \item[4a. ] As $\left[\frac{z}{s}\right]=\left[\frac{7}{4}\right]$, it holds that $i=1$
  corresponding to case a) of theorem \ref{T2}. Hence it must hold that
  \[ \beta^i (k,12) \equiv \beta^1 (k,12) = \beta^1 (1,12) - (k-1), \, k=1,\dots , 4 \]
  what is true according to (\ref{beta12}).
 \item[4b. ] As $\left[\frac{z}{s}\right] + 2 = 3$, it holds that $i=3$
  corresponding to case b) of theorem \ref{T2}. Hence it must hold that
  \[ \beta^3 (k,12) = \beta^3 (1,12) + (k-1), \, k=1,\dots , 4 \]
  what is true according to (\ref{beta12}).
 \item[4c. ] As $\left[\frac{z}{s}\right] + 1 = 2$, it holds that $i=2$
  corresponding to case c) of theorem \ref{T2}.
  
  As $i$ is even, it must hold that
  \begin{equation}\label{eq15-2}
   \beta^2 (k,12) = \beta (4+1-k,4), \, k=1,\dots , 4
  \end{equation}
  being $\beta(k,4)$ the $k$-th element of some next visiting permutation $\beta_4$.
  After replacing $4+1-k$ by $k$ in (\ref{eq15-2}), then it can be rewritten as
  \begin{equation}\label{eq15-3}
   \beta^2 (4+1-k,12) = \beta (k,4), \, k=1,\dots , 4
  \end{equation}
  being $\beta(k,4)$ the $k$-th element of some next visiting permutation $\beta_4$.
  By substituting (\ref{beta12}) in (\ref{eq15-3}) it yields that
  \[ \beta^2 (4,12) = 4 \equiv \beta (1,4) \]
  \[ \beta^2 (3,12) = 1 \equiv \beta (2,4) \]
  \[ \beta^2 (2,12) = 2 \equiv \beta (3,4) \]
  \[ \beta^2 (1,12) = 3 \equiv \beta (4,4) \]
  That is $\beta_4 \equiv ( 4 \, 1 \, 2 \, 3 )$ which is an admissible orbit of symbolic sequence CRLL.

\end{enumerate}

\end{enumerate}
\end{ex}

After obtaining $\beta_4$ in the decomposition $\beta_{12} = \beta_3 \circ \beta_4$
 we are going to use corollary \ref{cor1} to obtain $\beta_3$.

As $\left[\frac{z}{s}\right] = \left[\frac{7}{4}\right] = 1$ it results from corollary
 \begin{equation*} \beta(i,h)=\left\{
  \begin{array}{cl}
  \dfrac{\beta((i-1)4+1,12)}{4} & \hbox{ $i=1$} \\
&\\ 1  & \hbox{ $i= 2$} \\ & \\
  \dfrac{\beta((i-1)4+1,12)+ (4-1)  }{4} & \hbox{$i=3$}
  \end{array}
\right.\end{equation*}
and, by using (\ref{beta12}), it finally results in
  \[ \beta (1,3) = 3 \]
  \[ \beta (2,3) = 1 \]
  \[ \beta (3,3) = 2 \]
Therefore, $O_{h=3}$ has the next visiting permutation $(3\, 1\, 2)$ corresponding to
the orbit of symbolic sequence CRL.

Consequently, we conclude that $O_{12}$ is the composition of an orbit $O_{h=3}$ with an orbit $O_{s=4}$.

\section{Conclusion}

If we had a compound $hs$-periodic orbit we could decompose it in two orbits
of periods $h$ and $s$, respectively, according to theorem \ref{T2}. This process
is the opposite to that described in \cite{Mar1} where two orbits with periods $h$
and $s$ were composed to generate an $hs$-periodic orbit. Therefore, theorem
\ref{T2} closes the theoretical frame of composition and decomposition.

Theorem \ref{T2} states the necessary and sufficient conditions for the decomposition in simpler orbits. Meanwhile, corollary \ref{cor1} provides an algorithm for the computation of those orbits.

The decomposition theorems treated in this paper have an immediate application
(through Poincaré section) to those continuous physical systems showing bifurcation
diagrams similar to the one in figure \ref{figura1}.

Two periodic orbits (with $h$ and $s$ points in their respective Poincaré sections) can be composed into another periodic orbit, having $hs$ points in their Poincaré map in accordance with the composition theorem in \cite{Mar1} (or Lemma \ref{L1}). Now the opposite result can be also achieved using Theorem \ref{T2}.

An $s$-periodic orbit inside the $h$-periodic window, must follow a visiting order in its Poincaré map that can be decomposed using decomposition Theorem \ref{T2}: from a known periodic orbit, another two unique orbits can be described. This link between periodic orbits (not only from simpler to more complex, as studied in \cite{Mar1}, but also from complex to simpler orbits, as studied in this paper) imposes strong restrictions to a physical system dependent on one control parameter, whose underlying origin must be studied.

\section{Appendix}

\begin{thm}\label{ta}Let $O_q$ be an supercycle  of a $C^2-$unimodal map $f$ with the next visiting permutation $\beta_{q}=(\beta(1,q)\, \beta(2,q) \,\ldots  \beta(q,q))$. If $\beta_{q}$ is given by
\begin{itemize}
  \item [a)] \,  if \, $i=1, \ldots, [\frac{z}{s}]$
  \[ \beta^i(k,q)=
  \beta^i(1,q)-(k-1)
\]
  \item [b)] \,  if \, $i=[\frac{z}{s}]+2, \ldots, h$
  \[ \beta^i(k,q)=
  \beta^i(1,q)+(k-1)\]
  \item [c)] \,  if \, $i=[\frac{z}{s}]+1$
\[   \beta^i{(k,q)}=   \left\{                   \begin{array}{ll}
                     \beta{(k,s)}  & \hbox{if $i=[\frac{z}{s}]+1$ is odd} \\
                      \beta {(s+1-k,s)} & \hbox{if $i=[\frac{z}{s}]+1$ is even}
                    \end{array}\right.
\]
 \end{itemize}
for all $k=1,\ldots, s$.

\vspace{1em}

 Then for each $i=1, \ldots, h$ there exists only one $j \in \{1,\ldots, h\}$ such that
 \[\beta_q(H_i)=\{\beta ((i-1)s+k,q): \, k=1,\ldots, s\}=\{(j-1)s+r: \, r=1\ldots, s\}= H_j\]
 Furthermore
  \[\bigcup_{i=1}^h \beta_q(H_i)=\bigcup_{j=1}^h H_j=\{1, \ldots, hs \}\]

\end{thm}

\begin{proof} ~

\begin{enumerate}

\item Let $i=\left[ \frac{z}{s} \right] + 1$. From $c)$ it results that $\beta_q(H_i)=\beta_q(H_{\left[ \frac{z}{s} \right] + 1})=H_1$.

\item Let $i \neq \left[ \frac{z}{s} \right] + 1$. The proof is by contradiction. We suppose that $\beta_q(H_i) \neq H_j$, $j=1,\ldots,h$. \\
Let $i < \left[ \frac{z}{s} \right] + 1$ (for $i > \left[ \frac{z}{s} \right] + 1$ the proof is similar). As $\beta_q(H_i) \neq H_j$ and $\beta_q$ maps $s$ consecutive elements to $s$ consecutive elements (see item a) in theorem \ref{T2}), it results
\[ \left( \beta(i-1) s + 1, \, q \right) \neq \dot{s} \]
and
\[ \left( \beta(i-1) s + s, \, q \right) \neq \dot{s} \]
(where $\dot{s}$ denotes a multiple of $s$), consequently
\begin{equation}\label{apend2}
 \left( \beta(i-1) s + 1, \, q \right) = n s + r, \;\; r<s, \; n,r \in \mathbb{N}
\end{equation}
Taking into account (\ref{apend2}), item a) in theorem \ref{T2}, and definition \ref{d3} it results
\begin{equation}\label{apend3}
 \left( \beta(i-1) s + s, \, q \right) = n s + r - (s - 1),  \;\; r<s, \; n,r \in \mathbb{N}
\end{equation} 
As $\beta_q$ takes every value in $\{1,\,2,\,\dots, h s\}$, it results from (\ref{apend2}) and (\ref{apend3}) that 
\[\{1,\ldots, hs\}= A\bigcup \beta_q(H_i) \bigcup B \]
where
\begin{multline*}
     A=\{1,\ldots, ns+r-(s-1)-1\} \\  \beta_q(H_i)=\{ ns+r-(s-1), \ldots ns+r\} \\  B=\{ns+r+1,\ldots, hs\}
\end{multline*}

Notice that the cardinality of the sets $A$ and $B$ are respectively $(n-1)s+r$ and $(h-n)s-(r-1)$. Except for $H_i$ the images of the other boxes will be mapped into $s$ consecutive elements either in $A$ or in $B$ (see a) and b) in theorem \ref{ta}). Consequently the elements of $A$ and $B$ will be exhausted but for $r$ elements in $A$ and $s-(r-1)$ in $B$, therefore, the image of some boxes will not be formed by consecutive elements, which is in contradiction with the definition of $\beta_q$.

\end{enumerate} 

From items 1 and 2 above, it results 
  \[\bigcup_{i=1}^h \beta_q(H_i)=\bigcup_{j=1}^h H_j=\{1, \ldots, hs \}\]
where it has been taken into account that, as $\beta_q$ is a permutation, for every $i=1,\,\ldots,\,h$ there exists only one $j=1,\,\ldots,\,h$ such that $\beta_q(H_i)=H_j$.

\end{proof}


\begin{thebibliography}{99}

\bibitem{Hao}
H. Bai-Lin,
\textit{Elementary Symbolic dynamics},
World Scientific, 1989.

\bibitem{metro}
N. Metropolis, M. L. Stein, and P. R. Stein
On Finite Limit Sets for Transformations on the Unit Interval,
\textit{J. Combinatorial Theory},
\textbf{15}:25--44, 1973.

\bibitem{Mil}
J. Milnor, W. Thurston
On iterated maps of the interval,
in: \textit{Dynamical systems},
Lecture notes in mathematics, vol. 1342. Springer, Berlin, 1988, pp. 465--563.

\bibitem{De}
B. Derrida, A.  Gervois, Y. Pomeau,
Iteration of endomorphisms on the real axes and representation of numbers,
\textit{Annales de l'institut Henri Poincaré (A)},
\textbf{29}:305--356, 1978.

\bibitem{ref_1}
M. J. Feigenbaum,
Quantitative universality for a class of nonlinear transformations,
\textit{J. Stat. Phys},
\textbf{19}:25--52, 1978.

\bibitem{ref_2}
M. J. Feigenbaum,
The universal metric properties for nonlinear transformations,
\textit{J. Stat. Phys},
\textbf{21}:669--706, 1979.

\bibitem{jes}
J. San Martín,
Intermittency cascade,
\textit{Chaos, Solitons \& Fractals},
\textbf{32}:816--831, 2007.

\bibitem{Mar}
J. San Martín, Mª J. Moscoso, A. González Gómez,
The universal cardinal ordering of fixed points,
\textit{Chaos, Solitons \& Fractals},
\textbf{42}:1996--2007, 2009.

\bibitem{Mar1}
J. San Martín, Mª.J. Moscoso, A. González Gómez,
Composition law of cardinal ordering permutations,
\textit{Physica D},
\textbf{239}:1135--1146, 2010.

\bibitem{Au}
D. Auerbach, P. Cvitanovic, J.P. Eckmann, G. Gunarathe and I. Procaccia,
Exploring chaotic motions through periodic orbits,
\textit{Phys. Rev. Lett.},
\textbf{58}:2387--2389, 1987.

\bibitem{Cv}
P. Cvitanovic,
Invariant measurements of strange sets in terms of cycles,
\textit{Phys. Rev. Lett.},
\textbf{61}:2729--2732, 1988.

\bibitem{Caroll}
T. L. Carroll,
Approximating chaotic time series through unstable periodic orbits,
\textit{Phys. Rev. E},
\textbf{59}:1615--1621, 1999.

\bibitem{ott}
E. Ott, C. Grebogi, and J.A. Yorke,
Controlling Chaos,
\textit{Phys. Rev. Lett.},
\textbf{64}:1196--1199 , 1990.

\bibitem{py}
K. Pyragas,
Continuous Control of Chaos by Self-Controlling Feedback,
\textit{Phys. Lett. A},
\textbf{170}:421--427, 1992.

\bibitem{gutz}
J. Xu,  K.W. Chung
\textit{ Chaos in Classical and Quantum Mechanics},
Springer, New York, 1990.

\bibitem{artu}
R. Artuso, E. Aurell and P. Cvitanovic,
Recycling of Strange Sets: I. Cycle Expansions,
\textit{Nonlinearity},
\textbf{3}: 325--359, 1990.


\end{thebibliography}
\end{document}